\documentclass{llncs}
\usepackage{graphicx} \DeclareGraphicsExtensions{.eps}
\usepackage{amsmath,  amsfonts, amscd}
\usepackage{temporal}
\usepackage{epsfig}
\usepackage{color}
\usepackage{times}  
\newcommand{\x}{\mbox{r}}
\newcommand{\nx}{\bar{\mbox{r}}}
\newcommand{\y}{\mbox{c}}
\newcommand{\ny}{\bar{\mbox{c}}}
\newcommand{\z}{\mbox{g}}
\newcommand{\nz}{\bar{\mbox{g}}}

\newcommand{\fplay}{\tau}

 \newcommand{\set}[1]{\{#1\}}
\newcommand{\sseq}{s_0 s_1 s_2\ldots}
\newcommand{\fseq}{s_0 s_1 \ldots s_k}

\newcommand{\init}{s_I}
\newcommand{\qinit}{q_I}
\newcommand{\PA}{1}
\newcommand{\PB}{2}
\newcommand{\SA}{{S_1}}

\newcommand{\SB}{{S_2}}
\renewcommand{\EB}{{E_2}}

\newcommand{\SR}{S_{P}}

\newcommand{\gamegraph}{G}

\newcommand{\winsure}[1]{\langle \! \langle #1 \rangle \! \rangle_{\mathit{sure}}  }
\newcommand{\winas}[1]{\langle \! \langle #1 \rangle\! \rangle_{\mathit{almost}} }

\newcommand{\was}{\winsure{1}}
\newcommand{\waa}{\winas{1}}

\newcommand{\wbs}{\winsure{2}}
\newcommand{\wba}{\winas{2}}
\newcommand{\outcome}{\mathrm{Outcome}}
\newcommand{\Prb}{\mathrm{Pr}}

\newcommand{\trans}{\delta}
\newcommand{\distr}{{\cal D}}
\newcommand{\supp}{\mathrm{Supp}}

\newcommand{\Aa}{{\cal E}}

\newcommand{\slopefrac}[2]{\leavevmode\kern.1em
  \raise .5ex\hbox{\the\scriptfont0 #1}\kern-.1em
  /\kern-.15em\lower .25ex\hbox{\the\scriptfont0 #2}}

\newcommand{\pat}{\pi} 
\newcommand{\Pat}{\Pi} 
\newcommand{\Paths}{\Pi}
\newcommand{\straa}{\alpha} \newcommand{\Straa}{\mathcal{A}}
\newcommand{\strab}{\beta} \newcommand{\Strab}{\mathcal{B}}

\newcommand{\Inf}{\mathrm{Inf}} 
 \newcommand{\Reach}[1]{\mathrm{Reach}(#1)}
\newcommand{\Safe}[1]{\mathrm{Safe}(#1)}

\newcommand{\Buchi}[1]{\mathrm{Buchi}(#1)}
\newcommand{\coBuchi}[1]{\mathrm{coBuchi}(#1)}

\newcommand{\Parity}[1]{\mathrm{Parity}(#1)}
\newcommand{\coParity}[1]{\mathrm{coParity}(#1)}

\newcommand{\wh}{\widehat}
\newcommand{\wt}{\widetilde}
\newcommand{\Source}{\mathrm{Source}}
\newcommand{\AssumeLive}{\mathsf{AssumeFair}}
\newcommand{\AssumeSafe}{\mathsf{AssumeSafe}}
\newcommand{\AssRed}{\mathsf{AssRed}}

\newcommand\cS{{\cal S}}
\newcommand\cT{{\cal T}}

\newcommand\cO{{\cal O}}
\newcommand\cI{{\cal I}}
\newcommand\AP{P} 
\newcommand{\abc}{\Sigma}

\newcommand{\abcAP}{\Sigma}
\newcommand{\odd}{\mbox{odd}}
\newcommand{\even}{\mbox{even}}

\renewcommand{\implies}{\rightarrow}

\newcommand{\inp}{\tt{in}}
\newcommand{\out}{\tt{out}}

\newcommand{\req}{\tt{req}}
\newcommand{\cancel}{\tt{cancel}}
\newcommand{\grant}{\tt{grant}}

\newcommand{\game}{{\cal G}}

\newcommand{\lang}{L}
\newcommand{\word}{w} 

\renewcommand{\r}{)}
\renewcommand{\l}{(}

\newcommand{\safenv}{{E_s}}
\newcommand{\livenv}{{E_l}}

\newcommand{\prefixes}{\mbox{prefixes}}
\newcommand{\safety}{\mbox{safety}}

\newcommand{\cond}[1]{#1}

\usepackage{ifpdf}
\ifpdf
\else
\fi

\newenvironment{mydefinition}
{
  \smallskip
}
{
  \smallskip
}

\newenvironment{myexample}
{
  \vspace{-2mm}
  \begin{example}
}
{
  \end{example}
  \vspace{-2mm}
}

\title{Environment Assumptions for Synthesis}

\author{Krishnendu Chatterjee\inst{2} \and Thomas A. Henzinger\inst{1} \and
  Barbara Jobstmann\inst{1}} \institute{EPFL, Lausanne \and University of
  California, Santa Cruz}

\begin{document}
\maketitle

\begin{abstract}
  The synthesis problem asks to construct a reactive finite-state
  system from an $\omega$-regular specification. Initial
  specifications are often unrealizable, which means that there is no
  system that implements the specification.  A common reason for
  unrealizability is that assumptions on the environment of the system
  are incomplete.  We study the problem of correcting an unrealizable
  specification $\varphi$ by computing an environment assumption
  $\psi$ such that the new specification $\psi\rightarrow\varphi$ is
  realizable. Our aim is to construct an assumption $\psi$ that
  constrains only the environment and is as weak as possible.  We
  present a two-step algorithm for computing assumptions.  The
  algorithm operates on the game graph that is used to answer the
  realizability question.  First, we compute a safety assumption that
  removes a minimal set of environment edges from the graph.  Second,
  we compute a liveness assumption that puts fairness conditions on
  some of the remaining environment edges.  We show that the problem
  of finding a minimal set of fair edges is computationally hard, and
  we use probabilistic games to compute a locally minimal fairness
  assumption.
\end{abstract}

\section{Introduction}

Model checking has become one of the most successful verification techniques 
in hardware and software design.  
Although the methods are automated, the success of a verification process 
highly depends on the quality of the specifications.
Writing correct and complete specifications is a tideous task: 
it usually requires several iterations until a satisfactory specification is 
obtained.  
Specifications are often too weak 
(e.g., they may be vacuously satisfied \cite{Beer97,Kupfer99b});
or too strong 
(e.g., they may allow too many environment behaviors), 
resulting in spurious counterexamples.
In this work we automatically strengthen the environment constraints within  
specifications whose assumptions about the environment behavior are so weak 
as to make it impossible for a system to satisfy the specification.

Automatically deriving environment assumptions has been studied from 
several points of view.  
For instance, in circuit design one is interested in automatically 
constructing environment models that can be used in test-bench generation 
\cite{Vera,Specman}. 
In compositional verification, environment assumptions have been generated 
as the weakest input conditions under which a given software or hardware 
component satisfies a given specification \cite{Beyer07,Bobaru08}.
We follow a different path by leaving the design completely out of the
picture and deriving environment assumptions from the specification alone.  
Given a specification, we aim to compute a least restrictive environment
that allows for an implementation of the specification.

The assumptions that we compute can assist the designer in different ways.
They can be used as baseline necessary conditions in component-based model 
checking.  
They can be used in designing interfaces and generating test cases for 
components before the components themselves are implemented.
They can provide insights into the given specification.  
And above all, in the process of automatically constructing an 
implementation for the given specification (``synthesis''), they can be used 
to correct the specification in a way that makes implementation possible.

While specifications of closed systems can be implemented if they are 
\emph{satisfiable}, specifications of open systems can be implemented if 
they are \emph{realizable}
---i.e., there is a system that satisfies the specification without 
constraining the inputs. 
The key idea of our approach is that given a specification, if it is
not realizable, cannot be complete and has to be weakened by introducing 
assumptions on the environment of the system.
Formally, given an $\omega$-regular specification $\varphi$ which is not 
realizable, we compute a condition $\psi$ such that the new specification 
$\psi\rightarrow\varphi$ is realizable. 
Our aim is to construct a condition $\psi$ that does not constrain the 
system and is as weak as possible.  
The notion that $\psi$ must constrain only the environment can be captured 
by requiring that $\psi$ itself is realizable for the environment
---i.e., there exists an environment that satisfies $\psi$ without 
constraining the outputs of the system
(in general, in a closed loop around system and environment
---or controller and plant--- 
both $\psi$ and $\varphi$ refer to inputs as well as outputs).  

The notion that $\psi$ be as weak as possible is more difficult to capture.
We will show that in certain situations, there is no unique weakest 
environment-realizable assumption $\psi$, and in other situations, it is 
NP-hard to compute such an assumption.  

\smallskip\noindent{\bf Example.}  During our efforts of formally
specifying certain hardware designs \cite{Bloem07,Bloem07b}, several
unrealizable specifications were produced.  One specification was
particular difficult to analyze.  Its structure can be simplified to
the following example.
Consider a reactive system with the signals $\req$, $\cancel$, and
$\grant$, where $\grant$ is the only output signal.  The specification
requires that (i) every request is eventually granted starting from
the next time step, written in linear temporal logic as
$\always(\req\implies\nextt\eventually\grant)$; and (ii) whenever the
input $\cancel$ is received or $\grant$ is high, then $\grant$ has to
stay low in the next time step, written
$\always((\cancel\vee\grant)\implies \nextt\neg\grant)$.  This
specification is not realizable because the environment can force, by
sending $\cancel$ all the time, that the $\grant$ signal has to stay
low forever (Part~(ii)).  If $\grant$ has to stay low, then a request
cannot be answered and Part~(i) of the specification is violated.
One assumption that clearly makes this specification realizable is
$\psi_1=\always(\neg\cancel)$.  This assumption is undesirable because
it completely forbids the environment to send $\cancel$.  A system
synthesized with this assumption would simply ignore the signal
$\cancel$.
Assumption $\psi_2=\always(\eventually(\neg\cancel))$ and
$\psi_3=\always(\req\implies\eventually(\neg\cancel))$ are more
desirable but still not satisfactory: $\psi_2$ forces the environment
to lower $\cancel$ infinitely often even when no requests are send and
$\psi_3$ is not strong enough to implement a system that in each step
first produces an output and then reads the input: assume the system
starts with output $\grant=0$ in time step~0, then receives the input
$\req=1$ and $\cancel=0$, now in time step~1, it can choose between
(a) $\grant=1$, or (b) $\grant=0$.  If it chooses to set grant to high
by (a), then the environment can provide the same inputs once more ($\req=1$
and $\cancel=0$) and can set all subsequent inputs to $\req=0$ and
$\cancel=1$.  Then the environment has satisfied $\psi_3$ because
during the two requests in time step~0 and~1 $\cancel$ was kept low
but the system cannot fulfill Part~(i) of its specification without
violating Part~(ii) due to $\grant=1$ in time step~1 and $\cancel=1$
afterwards.  On the other hand, if the system decides to choose to set
$\grant=0$ by (b), then the environment can choose to set the inputs to
$\req=0$ and $\cancel=1$ and the system again fails to fulfill
Part~(i) without violating~(ii).
The assumption
$\psi_4=\always(\req\implies\nextt\eventually(\neg\cancel))$, which is
a subset of $\psi_3$, is sufficient.  However, there are infinitely
many sufficient assumptions between $\psi_3$ and $\psi_4$, e.g,
$\psi_3'=(\neg\cancel \wedge\nextt(\psi_3)) \vee \psi_3$.
The assumption
$\psi_5=\always(\req\implies\nextt\eventually(\neg\cancel \vee
\grant))$ is also weaker then $\psi_3$ and still sufficient because
the environment only needs to lower $\cancel$ eventually if a request
has not been answered yet.
Finally, let $\xi= \req\implies\nextt\eventually(\neg\cancel \vee
\grant)$, consider the assumption $\psi_6 = \xi \weakuntil
(\xi\wedge(\cancel\vee \grant)\wedge\nextt\grant)$, which is a
sufficient assumption.  It is desirable because it states that
whenever a request is send the environment has to eventually lower
$\cancel$ if it has not seen a $\grant$, but as soon as the system
violates its specification (Part~(ii)) all restrictions on the
environment are dropped.  If we replace $\xi$ in $\psi_6$ with
$\xi'=\req\implies\eventually(\neg\cancel \vee \grant)$, we get again
an assumption that is not sufficient for the specification to be
realizable.
This example shows that the notion of weakest and desirable are hard
to capture.

\smallskip\noindent{\bf Contributions.}
The realizability problem (and synthesis problem) can be reduced to emptiness 
checking for tree automata, or equivalently, to solving turn-based two-player 
games on graphs.
More specifically, an $\omega$-regular specification $\varphi$ is realizable 
iff there exists a winning strategy in a certain parity game constructed 
from $\varphi$.
If $\varphi$ is not realizable, then we construct an environment assumption 
$\psi$ such that $\psi\rightarrow\varphi$ is realizable, in two steps.
First, we compute a safety assumption that removes a minimal set of 
environment edges from the graph graph.  
Second, we compute a liveness assumption that puts fairness conditions on
some of the remaining environment edges of the game graph:
if these edges can be chosen by the environment infinitely often, then they 
need to be chosen infinitely often.
While the problem of finding a minimal set of fair edges is shown to be 
NP-hard, a local minimum can be found in polynomial time (in the size of 
the game graph) for 
B\"uchi and co-B\"uchi specifications, and in NP $\cap$ coNP for parity 
specifications. 
The algorithm for checking the sufficiency of a set of fair edges is of 
independent theoretical interest, as it involves a novel reduction of 
deterministic parity games to probabilistic parity games.

We show that the resulting conjunction of safety and liveness assumptions 
is sufficient to make the specification realizable, and itself realizable by 
the environment.
We also illustrate the algorithm on several examples, showing that it 
computes natural assumptions.

\smallskip\noindent{\bf Related works.} There are some related works that 
consider games that are not winning, methods of restricting the
environment, and constructing most general winning strategies in games.
The work of~\cite{Faella07} considers games that are not winning, and 
considers \emph{best-effort} strategies in such games.
However, relaxing the winning objective to make the game winning is not 
considered.
In~\cite{CHJ06}, a notion of non-zero-sum game is proposed, where the 
strategies of the environment are restricted according to a given objective, 
but the paper does not study how to obtain an environment objective that is 
sufficient to transform the game to a winning one.
A minimal assumption on a player with an objective can be captured by the 
most general winning strategy for the objective. 
The result of~\cite{Bernet02} shows that such most general winning 
strategies exist only for safety
games, and also presents an approach to compute a strategy, called a 
\emph{permissive strategy}, that subsumes behavior of all memoryless 
winning strategies.
Our approach is different, as it attempts to construct the minimal assumption
for the environment that makes the game winning, and we derive 
assumptions from the specification alone.

\smallskip\noindent{\bf Outline.}
In Section~\ref{sec:preliminaries}, we introduce the necessary
theoretical background for defining and computing environment
assumptions.  
Section~\ref{sec:assumptions} discusses environment assumptions 
and why they are difficult to capture.
In Section~\ref{sec:safety} and~\ref{sec:liveness}, we compute, 
respectively, safety and liveness assumptions, which are then 
combined in Section~\ref{sec:combining}.  

\section{Preliminaries}
\label{sec:preliminaries}

\noindent{\bf Words, Languages, Safety, and Liveness.}
Given a finite alphabet $\Sigma$ and an infinite word $w\in \Sigma^{\omega}$, 
we use $w_i$ to denote the $(i+1)^{th}$ letter of $w$, and $w^i$ to denote 
the finite prefix of $w$ of length~$i+1$.  
Note that the first letter of a word has index~$0$. 
Given a word $w\in\Sigma^\omega$, we write $\even(w)$ for the subsequence of 
$w$ consisting of the even positions ($\forall i \ge 0:\even(w)_i=w_{2i}$).  
Similarly, $\odd(w)$ denotes the subsequence of the odd positions.
Given a set $L\subseteq \Sigma^{\omega}$ of infinite words, we define the set 
of finite prefixes by 
$\prefixes(L)= \{v \in \Sigma^* \mid \exists w \in L, i \ge 0: v=w^i\}$.  
Given a set $L\subseteq \Sigma^*$ of finite words, we define the set 
of infinite limits by 
$\safety(L) = \{ w \in \Sigma^{\omega} \mid \forall i\ge 0: w^i \in L\}$.
We consider languages of infinite words.
A language $L\subseteq \Sigma^{\omega}$ is a \emph{safety} language if 
$L=\safety(\prefixes(L))$.  
A language $L\subseteq \Sigma^{\omega}$ is a \emph{liveness} language if 
$\prefixes(L)=\Sigma^*$.
Every language $L \subseteq \Sigma^{\omega}$ can be presented as the 
intersection of the safety language $\safety(\prefixes(L))$
and the liveness language $\safety(\prefixes(L))\backslash L$.

\smallskip\noindent{\bf Transducers.} 
We model reactive systems as deterministic finite-state transducers.
We fix a finite set $\AP$ of atomic propositions, and a partition of 
$\AP$ into a set $O$ of output propositions and a set $I$ of input 
propositions.  
We use the corresponding alphabets $\abcAP=2^{\AP}$, $\cO=2^O$, and 
$\cI=2^I$.
A \emph{Moore transducer} with input alphabet $\cI$ and output alphabet
$\cO$ is a tuple $\cT=\l Q, q_I, \delta, \kappa\r$, where $Q$ is a finite 
set of states, $q_I \in Q$ is the initial state, 
$\delta$: $Q \times \cI \rightarrow Q$ is the transition function, and
$\kappa$ is a state labeling function $\kappa$: $Q \rightarrow \cO$.
A \emph{Mealy transducer} is like a Moore transducer, except that 
$\kappa$: $Q\times \cI \rightarrow \cO$ is a transition labeling function.
A Moore transducer describes a reactive system that reads words over the
alphabet $\cI$ and writes words over the alphabet $\cO$.  
The environment of the system, in turn, can be described by a Mealy 
transducer with input alphabet $\cO$ and output alphabet $\cI$.
%
%
We extend the definition of the transition function $\delta$ to finite 
words $w \in \cI^*$ inductively by 
$\delta(q,w)=\delta(\delta(q,w^{|w|-1}),w_{|w|})$ for $|w|>0$.
Given an input word $w \in \cI^{\omega}$, the run of $\cT$ over $w$ is 
the infinite sequence $\pi \in Q^{\omega}$ of states such that $\pi_0 =q_I$, 
and $\pi_{i+1} = \delta(\pi_i, w_i)$ for all $i\ge 0$.
The run $\pi$ over $w$ generates the infinite word $\cT(w)\in\abc^{\omega}$ 
defined by $\cT(w)_i=\kappa(\pi_i)\cup w_i$ for all $i\ge 0$ in the case of 
Moore transducers;
and $\cT(w)_i=\kappa(\pi_i,w_i)\cup w_i$ for all $i\ge 0$ in the Mealy case.
The \emph{language} of the transducer $\cT$ is the set
$\lang(\cT)=\set{\cT(w)\mid w \in \cI^{\omega}}$
of infinite words generated by runs of~$\cT$.

\smallskip\noindent{\bf Specifications and Realizability.}
A \emph{specification} of a reactive system is an $\omega$-regular language 
$L\subseteq \Sigma^{\omega}$.
We use Linear Temporal Logic (LTL) formulae over the atomic proposition 
$\AP$, as well as $\omega$-automata with transition labels from $\abc$, to 
define specifications.
%
%
%
%
%
Given an LTL formula (resp.\ $\omega$-automaton) $\phi$, we write 
$\lang(\varphi)\subseteq\abc^\omega$ for the set of infinite words that 
satisfy (resp.\ are accepted by) $\varphi$.  
%
%
A transducer $\cT$ \emph{satisfies} a specification $L(\varphi)$, written 
$\cT \models \varphi$, if $\lang(\cT) \subseteq \lang(\varphi)$.  
Given an LTL formula (resp.\ $\omega$-automaton) $\varphi$, the 
\emph{realizability problem} asks if there exists a transducer 
$\cT$ with input alphabet $\cI$ and output alphabet $\cO$ such that 
$\cT\models\varphi$.
The specification $L(\varphi)$ is \emph{Moore realizable} if such a Moore 
transducer $\cT$ exists, and \emph{Mealy realizable} if such a Mealy
transducer $\cT$ exists.
Note that for an LTL formula, the specification $L(\varphi)$ is Mealy 
realizable iff $L(\varphi')$ is Moore realizable, where the LTL formula 
$\varphi'$ is obtained from $\varphi$ by replacing all occurrences of 
$o \in O$ by $\nextt o$.
The process of constructing a suitable transducer $\cT$ is called 
\emph{synthesis}.
The synthesis problem can be solved by computing winning strategies in 
graph games.

\smallskip\noindent{\bf Graph games.} 
We consider two classes of turn-based games on graphs, namely, 
two-player probabilistic games and two-player deterministic games.
The probabilistic games are not needed for synthesis, but we will use 
them for constructing environment assumptions.
For a finite set~$A$, a probability distribution on $A$ is a 
function $\trans$: $A\to[0,1]$ such that $\sum_{a \in A} \trans(a) = 1$. 
We denote the set of probability distributions on $A$ by $\distr(A)$. 
Given a distribution $\trans \in \distr(A)$, we write 
$\supp(\trans) = \{x \in A \mid \trans(x) > 0\}$ for the support
of $\trans$.
A \emph{probabilistic game graph} 
$\gamegraph =((S, E), (\SA,\SB,\SR),\trans)$ consists of a finite directed
graph $(S,E)$, a partition $(\SA$, $\SB$, $\SR)$ of the set $S$
of states, and a probabilistic transition function $\trans$: $\SR
\rightarrow \distr(S)$.
The states in $\SA$ are {\em player-$\PA$\/} states, where player~$\PA$ 
decides the successor state; 
the states in $\SB$ are {\em player-$\PB$\/} states, where player~$\PB$ 
decides the successor state; 
and the states in $\SR$ are {\em probabilistic\/} states, where the 
successor state is chosen according to the probabilistic transition
function.  
We require that for all $s \in \SR$ and $t \in S$, we have $(s,t) \in E$ iff 
$\trans(s)(t) > 0$, and we often write $\trans(s,t)$ for $\trans(s)(t)$.  
For technical convenience we also require that every state has at least 
one outgoing edge.
Given a subset $E' \subseteq E$ of edges, we write $\Source(E')$ for
the set $\{ s \in S \mid \exists t\in S: (s,t) \in E'\}$ of states 
that have an outgoing edge in $E'$.
The {\em deterministic game graphs} 
are the special case of the probabilistic game graphs
with $\SR = \emptyset$, that is, a deterministic game graph
$\gamegraph=((S,E),(\SA,\SB))$ consist of of a directed graph
$(S,E)$ together with the partition of the state space $S$ into 
player-1 states $\SA$ and player-2 states $\SB$.

\smallskip\noindent{\bf Plays and Strategies.}
An infinite path, or \emph{play}, of the game graph $\gamegraph$ is an 
infinite sequence $\pat=\sseq$ of states such that 
$(s_k,s_{k+1}) \in E$ for all $k\ge 0$.
We write $\Paths$ for the set of plays, and for a state $s \in S$, 
we write $\Paths_s\subseteq\Paths$ 
for the set of plays that start from~$s$.
A \emph{strategy} for player~$\PA$ is a function 
$\straa$: $S^*\cdot \SA \to S$ that for all finite sequences of states 
ending in a player-1 state (the sequence represents a prefix of a play), 
chooses a successor state to extend the play.
A strategy must prescribe only available moves, that is, 
$\straa(\fplay \cdot s) \in E(s)$ for all $\fplay \in S^*$ and $s \in \SA$.
The strategies for player~2 are defined analogously.
Note that we have only pure (i.e., nonprobabilistic) strategies.
We denote by $\Straa$ and $\Strab$ the set of strategies for player~$\PA$
and player~$\PB$, respectively.
A strategy $\straa$ is {\em memoryless\/} 
if it does not depend on the history of the play but only on the current 
state.
A memoryless player-1 strategy can be represented as a function 
$\straa$: $\SA \to S$, and a memoryless player-2 strategy 
is a function $\strab$: $\SB \to S$.
We denote by $\Straa^{M}$ and $\Strab^{M}$ 
the set of memoryless strategies for player~1 and player~2, 
respectively.

Once a starting state $s \in S$ and strategies $\straa \in \Straa$ and
$\strab \in \Strab$ for the two players are fixed, the outcome of the
game is a random walk $\pat_s^{\straa, \strab}$ for which the
probabilities of events are uniquely defined, where an \emph{event}
$\Aa \subseteq \Paths$ is a measurable set of plays.  
Given strategies
$\straa$ for player~1 and $\strab$ for player~2, a play $\pat=\sseq$
is \emph{feasible} if for all $k\ge 0$, we have 
$\straa(\fseq)=s_{k+1}$ if $s_k \in \SA$, and 
$\strab(\fseq)=s_{k+1}$ if $s_k \in \SB$.
Given two strategies
$\straa\in\Straa$ and $\strab\in\Strab$, and a state~$s\in S$, we
denote by $\outcome(s,\straa,\strab) \subseteq \Paths_s$ the set of
feasible plays that start from $s$.  
Note that for deterministic game graphs, the set 
$\outcome(s,\straa,\strab)$ contains a single play.
For a state $s \in S$ and an event $\Aa\subseteq\Paths$, we
write $\Prb_s^{\straa, \strab}(\Aa)$ for the probability that a play 
belongs to $\Aa$ if the game starts from the state $s$ and the two players
follow the strategies $\straa$ and~$\strab$, respectively.

\smallskip\noindent{\bf Objectives.}
An \emph{objective} for a player is a set $\Phi \subseteq \Pi$ of 
winning plays.
We consider $\omega$-regular sets of winning plays, which are 
measurable.
For a play $\pat = \sseq$, let $\Inf(\pat)$ be the set  
$\set{s \in S \mid \mbox{$s = s_k$ for infinitely many $k \geq 0$}}$
of states that appear infinitely often in~$\pat$.
\begin{enumerate}
\item
  \emph{Reachability and safety objectives.}
  Given a set $F \subseteq S$ of states, the reachability objective 
  $\Reach{F}$ requires that some state in $F$ be visited,
  and dually, 
  the safety objective $\Safe{F}$ requires that only states in $F$ 
  be visited.
  Formally, the sets of winning plays are
  $\Reach{F}= \set{\sseq \in \Pat \mid 
  \exists k \geq 0: s_k \in F}$
  and 
  $\Safe{F}=\set{\sseq \in \Pat \mid 
  \forall k \geq 0: s_k \in F}$.
\item
  \emph{B\"uchi and co-B\"uchi objectives.}
  Given a set $F \subseteq S$ of states, the B\"uchi objective 
  $\Buchi{F}$ requires that some state in $F$ be visited
  infinitely often, and dually, 
  the co-B\"uchi objective $\coBuchi{F}$ requires that only 
  states in $F$ be visited infinitely often.
  Thus, the sets of winning plays are
  $\Buchi{F}= \set{\pat \in \Pat \mid 
  \Inf(\pat)\cap F\neq\emptyset}$
  and 
  $\coBuchi{F}=\set{\pat \in \Pat \mid 
  \Inf(\pat)\subseteq F}$.
\item \emph{Parity objectives.}
  Given a function 
  $p$: $S \to \set{0,1,2,\ldots,d-1}$ 
  that maps every state to a \emph{priority}, 
  the parity objective $\Parity{p}$ requires that of the states 
  that are visited infinitely often, the least priority be even.
  Formally, the set of winning plays is 
  $\Parity{p}=\set{\pat\in\Pat \mid 
  \min\set{p(\Inf(\pat))} \text{ is even}}$. 
  The dual, co-parity objective has the set 
  $\coParity{p}=
  \set{\pat\in\Pat \mid \min\set{p(\Inf(\pat))} \text{ is odd}}$
  of winning plays.
\end{enumerate}
The parity objectives are closed under complementation: 
given a function $p$: $S \to \set{0,1,\ldots,d-1}$, consider the 
function $p+1$: $S \to \set{1,2,\ldots,d}$ defined by 
$p+1(s)=p(s)+1$ for all $s\in S$;
then $\Parity{p+1}=\coParity{p}$.
The B\"uchi and co-B\"uchi objectives are special cases of parity 
objectives with two priorities, namely, 
$p$: $S \to \set{0,1}$ for B\"uchi objectives with $F=p^{-1}(0)$, and 
$p$: $S \to \set{1,2}$ for co-B\"uchi objectives with $F=p^{-1}(2)$.
The reachability and safety objectives can be turned into B\"uchi 
and co-B\"uchi objectives, respectively, on slightly modified game graphs.

\smallskip\noindent{\bf Sure and Almost-Sure Winning.}
Given an objective~$\Phi$, a strategy $\straa\in\Straa$ is  
\emph{sure winning} for player~1 from a state $s\in S$ if for every 
strategy $\strab\in\Strab$ for player~2, we have
$\outcome(s,\straa,\strab) \subseteq \Phi$.
The strategy $\straa$ is \emph{almost-sure winning} for player~1 
from $s$ for $\Phi$ if for every player-2 
strategy~$\strab$, we have $\Prb_s^{\straa,\strab} (\Phi) =1$.
The sure and almost-sure winning strategies for player~2 are defined
analogously.
Given an objective~$\Phi$, the \emph{sure winning set} $\was(\Phi)$ 
for player~1 is the set of states from which player~1 has a sure winning 
strategy.
Similarly, the \emph{almost-sure winning set} $\waa(\Phi)$ for player~1 
is the set of states from which player~1 has an almost-sure winning 
strategy.
The winning sets $\wbs(\Phi)$ and $\wba(\Phi)$ for player~2 are defined 
analogously.
It follows from the definitions that for all probabilistic game
graphs and all objectives~$\Phi$, we have 
$\was(\Phi) \subseteq\waa(\Phi)$.
In general the subset inclusion relation is strict.
For deterministic games the notions of sure and almost-sure winning
coincide~\cite{Mar75}, that is, for all deterministic game graphs
and all objectives $\Phi$, we have $\was(\Phi)=\waa(\Phi)$, 
and in such cases we often omit the subscript.
Given an objective~$\Phi$, the \emph{cooperative winning set} 
$\winsure{1,2}(\Phi)$ is the set of states
$s$ for which there exist a player-1 strategy $\straa$ and a player-2
strategy $\strab$ such that $\outcome(s,\straa,\strab)\subseteq\Phi$.

\begin{theorem}[Deterministic games~\cite{EJ91}]\label{thrm:sure}
For all deterministic game graphs and parity objectives $\Phi$,
the following assertions hold: 
(i)~$\was(\Phi) = S \setminus \wbs(\Pi\setminus\Phi)$;
(ii)~memoryless sure winning strategies exist for both players from
their sure winning sets; 
and
(iii)~given a state $s\in S$, whether $s \in \was(\Phi)$ can be decided in 
NP $\cap$ coNP. 
\end{theorem}

\begin{theorem}[Probabilistic games~\cite{CJH03}]\label{thrm:almost}
Given a probabilistic game graph $G=((S,E),(\SA,\SB,\SR),\trans)$ 
and a parity objective $\Phi$ with $d$ priorities, 
we can construct a deterministic game graph 
$\wh{G}=((\wh{S},\wh{E}),(\wh{\SA},\wh{\SB}))$ with $S\subseteq\wh{S}$,
and a parity objective 
$\wh{\Phi}$ with $d+1$ priorities 
such that 
(i)~$|\wh{S}|=O(|S|\cdot d)$ and $|\wh{E}|=O(|E|\cdot d)$; and
(ii)~the set $\waa(\Phi)$ in $G$ is equal to the set $\was(\wh{\Phi}) \cap S$
in $\wh{G}$.
Moreover, memoryless almost-sure winning strategies exists for both players 
from their almost-sure winning sets in $G$.
\end{theorem}

\noindent{\bf Realizability Games.} 
The realizability problem has the following game-theoretic formulation.

\begin{theorem}[Reactive synthesis~\cite{Pnueli89}]
  \label{thm:synthesis}
  Given an LTL formula (resp.\ $\omega$-automaton) $\varphi$, 
  we can construct a deterministic game graph $\gamegraph$, a state 
  $\init$ of $G$, and a parity objective $\Phi$ such that $L(\varphi)$ is 
  realizable iff $\init\in\was(\Phi)$.
\end{theorem}

\noindent
The deterministic game graph $G$ with parity objective $\varphi$ referred 
to in Theorem~3 is called a \emph{synthesis} game for~$\varphi$.
Starting from an LTL formula $\varphi$, we construct the synthesis game 
by first building a nondeterministic B\"uchi automaton that accepts 
$L(\varphi)$~\cite{Vardi94}.  
Then, following the algorithm of \cite{Piterm06}, we translate this 
automaton to a deterministic parity automaton that accepts $L(\varphi)$.
By splitting every state of the parity automaton w.r.t.\ inputs $I$ and 
outputs $O$, we obtain the synthesis game.  
Both steps involve an exponential blowup that is unavoidable:
for LTL formulae $\varphi$, the realizability problem is 2EXPTIME-complete
\cite{Rosner92}.

Synthesis games, by relating paths in the game graph to the 
specification $\varphi$, have the following special form.  
A \emph{synthesis game} $\game$ is a tuple 
$(\gamegraph,\init,\lambda,\Phi)$, where $\gamegraph = ((S,E),(\SA,\SB))$
is a deterministic bipartite game graph, in which player-1 and player-2 
states strictly alternate 
(i.e., $E \subseteq (S_1\times S_2) \cup (S_2\times S_1)$),
the initial state 
$\init\in \SA$ is a player-1 state, the labeling function 
$\lambda$: $S\implies \cO \cup \cI$ maps player-1 and player-2 states to 
letters in $\cI$ and $\cO$, respectively
(i.e., $\lambda(s)\in\cI$ for all $s\in\SA$, and $\lambda(s)\in\cO$ for all 
$s\in\SB$), 
and $\Phi$ is a parity objective.  
Furthermore, synthesis games are deterministic w.r.t.\ input and output 
labels, that is, for all edges $(s,s'),(s,s'') \in E$, if 
$\lambda(s')=\lambda(s'')$, then $s'=s''$.  Without loss of
generality, we assume that synthesis games are complete w.r.t.\ input
and output labels, that is, for all state $s\in \SA$ ($\SB$) and
$l\in\cO$ ($\cI$, respectively), there exists an edge $(s,s')\in E$
such that $\lambda(s')=l$. 
We define a function $\word$: $\Paths \implies \abc^\omega$ that maps 
each play to an infinite word such that 
$w_i = \lambda(\pat_{2i+1}) \cup \lambda(\pat_{2i+2})$ for all $i\ge 0$.
Note that we ignore the label of the initial state.
Given a synthesis game $\game$ for a specification formula or 
automaton $\varphi$, every Moore transducer 
$\cT=\l Q, \qinit, \delta, \kappa\r$ that satisfies $L(\varphi)$ 
represents a winning strategy $\straa$ of player~1 as follows: 
for all sequences $\fplay \in (\SA\SB)^*\cdot S_1$, let $w$ be the finite 
word such that $w_i = \lambda(\fplay_{i+1})$ for all $0 \le i < |\fplay|$;
then, if there exists an edge $\l \fplay_{|\fplay|}, s' \r \in E$ with 
$\lambda(s')=\kappa(\delta(\qinit,\odd(w)))$, 
then $\straa(\fplay)=s'$, and otherwise $\straa(\fplay)$ is arbitrary.
Conversely, every memoryless winning strategy $\straa$ of player~1
represents a Moore transducer $\cT=\l Q, \qinit, \delta, \kappa\r$ that
satisfies $L(\varphi)$ as follows: 
$Q=\SA$, $\qinit=\init$, $\kappa(q)=\lambda(\straa(q))$, and 
$\delta(q,l)=s'$ if $\lambda(s')=l$ and $(\straa(q),s')\in E$.

\section{Assumptions}
\label{sec:assumptions}

%
%

In this section, we discuss about environment assumptions in general,  
illustrating through several simple examples, and then 
identify conditions that every assumption has to satisfy.

Given a specification $\varphi$ that describes the desired behavior of
an open system $\cS$, we search for assumptions on the environment of
$\cS$ that are sufficient
to ensure that $\cS$ exists and satisfies $\varphi$.  
The assumptions we study are independent of the actual implementation.  
They are derived from the given specification and can
be seen as part of a correct specification.  We first define what it
means for an assumption to be sufficient.

\begin{mydefinition}
  Let $\varphi$ be a specification.  A language $\psi \subseteq
  \abcAP$ is a \emph{sufficient environment assumption for $\varphi$}
  if $(\abc^\omega\setminus\psi) \cup \varphi$ is realizable.
\end{mydefinition}

\begin{myexample}
  \label{ex:until}
  Consider the specification $\varphi=\out \until \inp$.  There exists
  no system $\cS$ with input $\inp$ and output $\out$ such that $\cS
  \models \varphi$, because $\cS$ cannot control the value of $\inp$
  and $\varphi$ is satisfied only if $\inp$ eventually becomes true.
  We have to weaken the specification to make it realizable.  A
  candidate $\psi$ for the assumption is $\eventually \inp$ because it
  forces the environment to assert the signal $\inp$ eventually, which
  allows the system to fulfill $\varphi$.  Further candidates are
  $\false$, which makes the specification trivially realizable,
  $\nextt\inp$, which forces the environment to assert the signal
  $\inp$ in the second step, $\eventually\out$, or $\eventually
  \neg\out$.  The last two assumptions lead to new specifications of
  the following form
  $\varphi'=\psi\implies\varphi=\eventually\out\implies \varphi =
  \always(\neg\out) \vee\varphi$.  The system can implement $\varphi'$
  independent of $\varphi$ simply by keeping $\out$ low all the time.
\end{myexample}

Example~\ref{ex:until} shows that there are several assumptions that
allow to implement the specification but not all of them are
satisfactory.  For example, the assumption $\false$ does not provide the
desired information.  Similarly, the assumption $\eventually\out$ is
not satisfactory, because it cannot be satisfied by any environment
that controls $\inp$.  Intuitively, assumptions that are false or that
can be falsified by the system correspond to a new specification
$\psi\rightarrow\varphi$ that can be satisfied vacuously
\cite{Beer97,Kupfer99b} by the system.
In order to exclude those assumptions, we require that an assumption
fulfills the following condition:

\begin{itemize}
\item[\emph{(\cond{1})}] \emph{Realizable for the environment:} The
  system cannot trivially falsify the assumption, so there exists an
  implementation of the environment that satisfies $\psi$. Formally,
  $\psi$ is Mealy realizable\footnote{Note that we ask here for a
    Mealy transducer, since the system is a Moore transducer.}.
\end{itemize}

Note that Condition~\cond{1} induces that $\varphi$ has to be
satisfiable for $\psi$ to exist.
If $\varphi$ is not satisfiable there exists only the trivially
solution $\psi=\false$.  We assume from now on that $\psi$ is
satisfiable.  Apart from Condition~\cond{1}, we ask for a condition to
compare or order different assumptions.  We aim to restrict the
environment ``as little as possible''.  An obvious candidate for this
order is language inclusion:

\begin{itemize}
\item[\emph{(\cond{2})}] \emph{Maximum:}
  An assumption $\psi$ is maximal if there exists no other sufficient
  assumption that includes $\psi$.
  There is no language $\psi'\subseteq \abcAP$ such that
  $\psi\subset\psi'$ and $(\abc^\omega\setminus \psi') \cup \varphi$
  is realizable.
\end{itemize}

The following example shows that using language inclusion we cannot
ask for a unique maximal assumption.
\begin{myexample}
  \label{ex:unique}
  Consider the specification $\varphi = (\out \until \inp_1) \vee
  (\neg\out \until \inp_2)$, where $\inp_1$ and $\inp_2$ are inputs
  and $\out$ is an output.  Again, $\varphi$ is not realizable.
  Consider the assumptions $\psi_1=\eventually\inp_1$ and
  $\psi_2=\eventually\inp_2$.  Both are sufficient because assuming
  $\psi_1$ the system can keep the signal $\out$ constantly high and
  assuming $\psi_2$ it can keep $\out$ constantly low.  However, if we
  assume the disjunction $\psi=\psi_1\vee\psi_2$, the system does not
  know, which of the signals $\inp_1$ and $\inp_2$ the environment is
  going to assert eventually.
  Since a unique maximal assumption has to subsume all other
  sufficient assumptions and $\psi$ is not sufficient, it follows that
  there exists no unique maximal assumption that is sufficient.

\end{myexample}

Let us consider another example to illustrate the difficulties that
arise when comparing environment assumptions w.r.t.\ language
inclusion.

\begin{myexample}
  \label{ex:safety}
  Assume the specification $\varphi = \always(\inp\implies\nextt\out)
  \wedge \always(\out\implies(\nextt\neg\out))$ with input signal
  $\inp$ and output signal $\out$.  The specification is not
  realizable because whenever $\inp$ is set to true in two consecutive
  steps, the system cannot produce a value for $\out$ such that
  $\varphi$ is satisfied.
  One natural assumption is $\psi=\always(\inp
  \implies\nextt\neg\inp)$.  Another assumption is $\psi'=\psi \vee
  \eventually(\neg\inp\wedge \nextt\out)$, which is weaker than $\psi$
  w.r.t.\ language inclusion and still realizable.  Looking at the
  resulting system specification $\psi'\rightarrow \varphi = (\psi
  \vee \eventually(\neg\inp\wedge \nextt\out)) \implies \varphi = \psi
  \rightarrow (\always(\neg\inp\implies\nextt\neg\out) \wedge
  \varphi)$, we see that $\psi'$ restricts the system instead of the
  environment.
\end{myexample}

Intuitively, using language inclusion as ordering notion, results in
maximal environment assumptions that allow only a single
implementation for the system.  We aim for an assumption that does not
restrict the system if possible.
One may argue that $\psi$ should talk only about input signals.
Let us consider the specification of Example~\ref{ex:safety} once more.
Another sufficient assumption is $\psi''= (\inp \implies\nextt\neg\inp)
\weakuntil (\out \wedge \nextt \out)$, which is weaker than $\psi$.
Intuitively, $\psi''$ requires that the environment guarantees $(\inp
\implies\nextt\neg\inp)$ as long as the system did not make a mistake
(by setting $\out$ to true in two consecutive steps), which clearly
means the intuition of an environment assumption.  The challenge is to
find an assumption that (a) is sufficient, (b) does not restrict the
system, and (c) gives the environment maximal freedom.

Note that the assumptions $\psi$ and $\psi''$ are safety assumptions,
while the assumptions in Example~\ref{ex:unique} are liveness
assumptions.  In general, every language can be split into a safety
and a liveness component.  We use this separation to provide a way to
compute environment assumption that fulfills our criteria.

We consider restriction on game graphs of synthesis games to find
sufficient environment assumptions.  More precisely, we propose to put
restrictions on player-2 edges, because they correspond to decisions
the environment can make.  If the given specification is satisfiable,
this choice of restrictions leads to assumptions that fulfill
the realizability for the environment.

\section{Safety Assumptions}
\label{sec:safety}

In this section, we define and compute an assumption that restricts
the safety behavior of the environment.

\subsection{Non-restrictive Safety Assumption on Games}

\begin{mydefinition}
  Given a deterministic game graph $\gamegraph=\l\l S,E\r,\l \SA,\SB\r\r$
  and the winning objective $\Phi$ for player~1.
  A \emph{safety assumption} on the set $\safenv \subseteq \EB$ of
  edges requires that player~2 chooses only edges outside from~$\safenv$. 
  A synthesis game $\game=(\gamegraph,\init,\lambda)$ with a safety
  assumption on $\safenv$ defines an environment assumption
  $\psi_\safenv$ as the set of words $w \in \abc^\omega$ such that
  there exists a play $\pi \in \Paths_{\init}$ with $w=\word(\pi)$,
  where for all $i\ge 0$, we have $(\pi_i,\pi_{i+1})\not\in\safenv$.
\end{mydefinition}

The set $\safenv$ can be seen as a set of forbidden edges of player~2.
A natural order on safety assumptions is the number of edges in a
safety assumption.  We write $\safenv\le\safenv'$ if
$|\safenv|\le|\safenv'|$ holds.
A safety assumption refers to the safety component of a winning
objective, which can be formulated as $\Phi_S =
\Safe{\winsure{1,2}(\Phi)}$.
Formally, the winning objective of player~1 is modified to
$\AssumeSafe(\safenv,\Phi)=\{\pat=\sseq \mid$ either (i) there exists
$i\ge 0$ such that $(s_i,s_{i+1}) \in \safenv$, or (ii) $\pi \in
\Phi_S\}$ denoting the set of all plays in which either one of the
edges in $\safenv$ is chosen, or that satisfies the safety component
of $\Phi$.

\begin{mydefinition}
  Given a deterministic game graph $\gamegraph=\l\l S,E\r,\l
  \SA,\SB\r\r$, a winning objective $\Phi$ for player~1, and a safety
  assumption $\safenv$, the safety assumption on $\safenv$ is
  \emph{safe-sufficient for state} $s\in S$ and $\Phi$ if player~1 has
  a winning strategy from $s$ for the objective
  $\AssumeSafe(\safenv,\Phi)$.
\end{mydefinition}

\begin{theorem}
  \label{thm:safety}
  Let $\varphi$ be a specification and let
  $\game_\varphi=(\gamegraph,\init,\lambda)$ be a synthesis game for
  $\varphi$ with the winning objective $\Phi$.
  An environment assumption $\psi_\safenv$ defined by a safety
  assumption $\safenv$ on $\game_\varphi$ that is safe-sufficient for
  $\init$ and $\Phi$, is sufficient for
  $\varphi'=\safety(\prefixes(\lang(\varphi)))$.
  Note that if $\varphi$ is a safety language, then $\psi_\safenv$ is
  sufficient for $\varphi$.
\end{theorem}
\begin{proof}
  Since $\safenv$ is safe-sufficient for $\init$ and $\Phi$, player~1
  has a memoryless winning strategy $\straa$ for $\init$ and
  $\AssumeSafe(\safenv,\Phi)$.  We know from
  Theorem~\ref{thm:synthesis} that $\straa$ corresponds to a
  transducer $\cT$.
  We need to show that the language of the transducer $\lang(\cT)$ is
  a subset of the new specification $(\abc\setminus\psi_\safenv)
  \cup\safety(\prefixes(\lang(\varphi)))$.  A run of $\cT$ on a word
  $w \in \cI^\omega$ corresponds to a winning play $\pat \in
  \AssumeSafe(\safenv,\Phi)$ of $\game_\varphi$.
  A play $\pi'=s_0 s_1\dots \in\AssumeSafe(\safenv,\Phi)$ either has
  an edges $(s_i,s_{i+1})\in \safenv$, then $\word(\pi') \in
  (\abc\setminus\psi_\safenv)$, or $\pi'\in
  \Safe{\winsure{1,2}(\Phi)}$, then we have that $\word(\pi') \in$
  $\safety(\prefixes(\lang(\varphi)))$.
\qed
\end{proof}

In the following example, we show that there exist safety games such
that for some state $s$ there is no unique smallest assumption that is
safe-sufficient for $s$.

\begin{figure}[bt]
  \begin{minipage}[t]{0.48\textwidth}
    \centerline{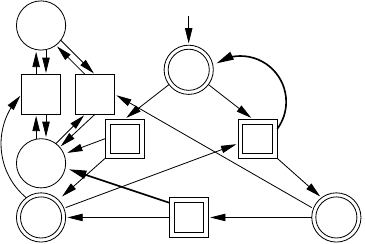} 
    \caption{Game with two equally small safe-sufficient assumptions
      for $s_1$: $\safenv=\set{(s_3,s_1)}$ and
      $\safenv'=\set{(s_5,s_7)}$.}
    \label{fig:unique}
  \end{minipage}
  \hspace{0.04\textwidth}
  \begin{minipage}[t]{0.48\textwidth}
    \centerline{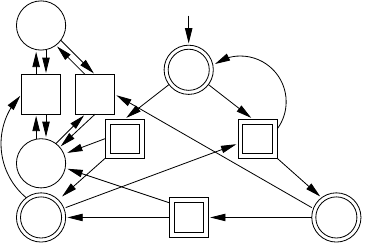} 
    \caption{Synthesis game for $\always(\inp\implies\nextt\out)
      \wedge\always(\out\implies\nextt\neg\out)$.}
    \label{fig:unique_phi}
  \end{minipage}
\end{figure}

\begin{myexample}
  Consider the game shown in Figure~\ref{fig:unique}.  Circles denote
  states of player~1, boxes denote states of player~2.  The winning
  objective of player~1 is to stay in the set
  $\set{s_1,s_2,s_3,s_4,s_5,s_6}$ denoted by double lines.  Player~1
  has no winning strategy for $s_1$.  There are two equally small
  safety assumptions that are safe-sufficient for $s_1$:
  $\safenv=\set{(s_3,s_1)}$ and $\safenv'=\set{(s_5,s_7)}$.  In both
  cases, player~1 has a winning strategy from state $s_1$.

  If we consider a specification, where the corresponding synthesis
  game has this structure, neither of these assumptions is
  satisfactory.
  Figure~\ref{fig:unique_phi} shows such a synthesis game for the
  specification $\always(\inp\implies\nextt\out)$ $\wedge$
  $\always(\out\implies\nextt\neg\out)$ with input signal $\inp$ and
  output signal $\out$ (cf.\ Example~\ref{ex:safety}).  Assuming the
  safety assumption $\safenv$, the corrected specification would allow
  only the single implementation, where $\out$ is constantly keep low.
  The second assumption $\safenv'$ leads to a corrected specification
  that additionally enforces
  $\always(\neg\inp\implies\nextt\neg\out)$.
\end{myexample}

Besides safe-sufficient, we also ask for an assumption that does not
restrict player~1.  This condition can be formulated as follows.
  Given a deterministic game graph $\gamegraph=\l\l S,E\r,\l
  \SA,\SB\r\r$, a winning objective $\Phi$ for player~1, and a safety
  assumption $\safenv$.  We call the safety assumption on $\safenv$
  \emph{restrictive} for state $s\in S$ and $\Phi$, if there exist
  strategies $\straa \in\Straa$ and $\strab\in\Strab$ of player~1
  and~2, respectively such that the play $\outcome(s,\straa,\strab)$
  contains an edge from $\safenv$ and is in $\Phi_S$.
  A non-restrictive safety assumption should allow any edge that does
  not lead to an immediate violation of the safety component of the
  winning objective of player~1.
\begin{theorem}
  \label{thm:minimal}
  Given a deterministic game graph $\gamegraph=\l\l S,E\r,\l \SA,\SB\r\r$,
  a winning objective $\Phi$ for player~1, and a state $s\in S$,
  if $s \in\winsure{1,2}(\Phi)$, then there exists a unique minimal
  safety assumption $\safenv$ that is non-restrictive and
  safe-sufficient for state $s$ and $\Phi$.
  Let $s \in\winsure{1,2}(\Phi)$ and let $\safenv$ be this unique
  minimal safety assumption for $s$ and $\Phi$, then player~2 has
  winning strategy for $s$ and the objective to avoid the edges in
  $\safenv$.
\end{theorem}

Applying this theorem to environment assumptions, we get the following
theorem.

\begin{theorem}
  \label{thm:env-realizable}
  Let $\varphi$ be a satisfiable specification and let
  $\game_\varphi=(\gamegraph,\init,\lambda)$ be a synthesis game for
  $\varphi$ with winning objective $\Phi$, then there exists a unique
  minimal safety assumption $\safenv$ that is non-restrictive and
  safe-sufficient for state $s$ and $\Phi$ and the corresponding
  environment assumption $\psi_\safenv$ is realizable for the
  environment.
\end{theorem}



\subsection{Computing Non-restrictive Safety Assumptions}

Given a deterministic game graph $\gamegraph$ and a winning objective $\Phi$, we
compute a non-restrictive safety assumption $\safenv$ as follows:
first, we compute the set $\winsure{1,2}(\Phi)$.  Note that for this set 
the players cooperate.  We can compute $\winsure{1,2}(\Phi)$ in polynomial
time for all objectives we consider. In particular, if $\Phi$ is a
parity condition, $\winsure{1,2}(\Phi)$ can be computed by reduction
to B\"uchi \cite{King01}.  The safety assumption $\safenv$ is the set
of all player-2 edges $(s,t)\in \EB$ such that $s\in
\winsure{1,2}(\Phi)$ and $t\not\in\winsure{1,2}(\Phi)$.

\begin{theorem}
  \label{thm:compute_safety}
  Consider  a deterministic game graph $\gamegraph$ with a winning objective $\Phi$.  
  The safety assumption $\safenv=\set{(s,t) \in \EB \mid s\in
  \winsure{1,2}(\Phi) \text{ and } t\not\in\winsure{1,2}(\Phi)}$  
  is the unique minimal safety assumption that is non-restrictive 
  and safe-sufficient for all states $s \in \winsure{1,2}(\Phi)$.
  The set $\safenv$ can be computed in polynomial time for all
  parity objectives $\Phi$.
\end{theorem}

For the game show in Figure~\ref{fig:unique}, we obtain the safety
assumption $\safenv=\set{(s_3,s_1),(s_5,s_7)}$.
For the corresponding synthesis game in Figure~\ref{fig:unique_phi},
$\safenv$ defines the environment assumption $\psi_\safenv=$
$(\neg\inp\vee\neg\out)\weakuntil 
((\neg\inp\vee\neg\out) \wedge 
(\inp \wedge \nextt(\neg\out)) \wedge 
(\out\wedge \nextt(\out)))$.
This safety assumption meets our intuition of a minimal environment
assumption, since it states that the environment has to ensure that
either $\inp$ or $\out$ is low as long as the system makes no obvious
fault by either violating $\always(\inp\implies\nextt\out)$ or
$\always(\out\implies\nextt\neg\out)$.

\section{Liveness Assumptions}
\label{sec:liveness}

\subsection{Strongly Fair Assumptions on Games}

\begin{mydefinition}
  Given a deterministic game graph $G=((S,E),(\SA,\SB))$ and the
  objective $\Phi$ for player~1,
  a \emph{strongly fair} assumption is a set $\livenv \subseteq \EB$ of
  edges requiring that player~2 plays in a way such that if a state $s
  \in \Source(\livenv)$ is visited infinitely often, then for all $t
  \in S$ such that $(s,t) \in \livenv$, the edge $(s,t)$ is chosen
  infinitely often.
\end{mydefinition}
  

This notion is formalized by modifying the objective of player~1 as
follows.  Let $\AssumeLive(\livenv,\Phi) = \{ \pat=\sseq \mid \text{
  either (i) } \exists(s,t) \in \livenv$, such that $s_k=s$ for
infinitely many $k$'s and $s_j=t$ for finitely many $j$'s, or (ii)
$\pat \in \Phi \}$ denote the set of paths $\pat$ such that either
(i)~there is a state $s\in \Source(\livenv)$ that appears infinitely
often in $\pat$ and there is a $(s,t)\in \livenv$ and $t$ appears only
finitely often in $\pat$, or (ii)~$\pat$ belongs to the objective
$\Phi$.  In other words, part (i) specifies that the strong fairness
assumption on $\livenv$ is violated, and part (ii) specifies that
$\Phi$ is satisfied.  The property that player~1 can ensure $\Phi$
against player-2 strategies respecting the strongly fair assumption on
edges is formalized by requiring that player~1 can satisfy
$\AssumeLive(\livenv,\Phi)$ against all player-2 strategies.

Given a deterministic game graph $G=((S,E),(\SA,\SB))$ and the
objective $\Phi$ for player~1, a strongly fair assumption on $\livenv
\subseteq \EB$ is \emph{sufficient} for state $s\in S$ and $\Phi$, if
player~1 has a winning strategy for $s$ for the objective
$\AssumeLive(\livenv,\Phi)$.
Furthermore, given a deterministic game graph $G=((S,E),(\SA,\SB))$
and the objective $\Phi$ for player~1, a state $s\in S$ is \emph{live
  for player~1}, if she has a winning strategy from $s$ for the
winning objective $\Safe{\winsure{1,2}(\Phi)}$.

\begin{theorem}
  \label{thm:complete}
  Given a deterministic game graph $G=((S,E),(\SA,\SB))$ and a
  Reachability, Safety, or B\"uchi objective $\Phi$.  If $s\in S$ is
  live for player~1, then there exists a strongly fair assumption
  $\livenv \subseteq \EB$ that is \emph{sufficient} for state $s\in S$
  and $\Phi$.
\end{theorem}

\begin{mydefinition}
  A synthesis game $\game=(\gamegraph,\init,\lambda)$ with a strongly
  fair assumption on $\livenv$ defines an environment assumption
  $\psi_\livenv$ as the set of words $w\in\abc^\omega$ such that there
  exists a play $\pat\in\Paths_{\init}$ with $w=\word(\pi)$ and for
  all edges $(s,t) \in \livenv$ either there exists $i\ge 0$ s.t.\ for
  all $j>i$ we have $\pi_i \ne s$, or there exist infinitely many $k$'s such
  that $\pi_k=s$ and $\pi_{k+1}=t$.  Note that this definition and the
  structure of synthesis games ensure that $\psi_\livenv$ is
  realizable.
\end{mydefinition}
These definitions together with Theorem~\ref{thm:synthesis}
and~\ref{thm:complete} lead to the following theorem.

\begin{theorem}
  \label{thm:liveness}
  Let $\varphi$ be a specification and $\game=(\gamegraph,\init,\lambda)$ a
  synthesis game for $\varphi$ with winning objective $\Phi$.
  If a strongly fair assumption on $\livenv$ is sufficient for $\init$
  and $\Phi$, then the environment assumption $\psi_\livenv$ is
  sufficient for $\varphi$ and realizable for the
  environment.  Furthermore, if $\Phi$ is a Reachability, Safety, or
  B\"uchi objective and $\init$ is live for player~1, then if there
  exists some sufficient assumption $\psi\ne\emptyset$, then there
  exists a strongly fair assumption that is sufficient.
\end{theorem}


\subsection{Computing Strongly Fair Assumptions}

We now focus on solution of deterministic player games with objectives
$\AssumeLive(\livenv,\Phi)$, where $\Phi$ is a parity objective.
Given a deterministic game graph $G$, an objective $\Phi$, and a strongly
fair assumption $\livenv$ on edges, we first observe that the
objective $\AssumeLive(\livenv,\Phi)$ can be expressed as an
implication: a strong fairness condition implies $\Phi$.  Hence given
$\Phi$ is a B\"uchi, coB\"uchi or a parity objective, the solution of
games with objective $\AssumeLive(\livenv,\Phi)$ can be reduced to
deterministic player Rabin games.  However, since deterministic Rabin games are
NP-complete we would obtain NP solution (i.e., a NP upper bound), even
for the case when $\Phi$ is a B\"uchi or coB\"uchi objective.  We now
present an efficient reduction to probabilistic games and show that
we can solve deterministic games with objectives
$\AssumeLive(\livenv,\Phi)$ in NP $\cap$ coNP for parity objectives
$\Phi$, and if $\Phi$ is B\"uchi or coB\"uchi objectives, then the
solution is achieved in polynomial time.

\medskip\noindent{\bf Reduction.}
Given a deterministic game graph $G=((S,E),(\SA,\SB))$, a parity objective $\Phi$ 
with a parity function $p$, and a set $\livenv \subseteq E_2$ of player-2 edges
we construct a probabilistic game 
$\wt{G}=((\wt{S},\wt{E}),(\wt{\SA},\wt{\SB},\wt{\SR}),\wt{\trans})$ 
as follows.
\begin{enumerate}
\item \emph{State space.} 
$\wt{S}= S \cup \set{\wt{s} \mid s \in \Source(\livenv), E(s) \setminus \livenv \neq \emptyset}$,
i.e., along with states in $S$, there is a copy $\wt{s}$ of a state $s$ in $\Source(\livenv)$ 
such that all out-going edges from $s$ are not contained in $\livenv$.

\item \emph{State space partition.} 
$\wt{\SA}=\SA$; $\wt{\SR}=\Source(\livenv)$; and 
$\wt{\SB}=\wt{S} \setminus (\wt{\SA} \cup \wt{\SR})$.
The player-1 states in $G$ and $\wt{G}$ coincide; every state
in $\Source(\livenv)$ is a probabilistic state and all other states 
are player-2 states.

\item \emph{Edges and transition.} We explain edges for the three
different kind of states.
\begin{enumerate} 
\item For $s \in \wt{\SA}$ we have $\wt{E}(s)=E(s)$, i.e., the set of edges from
player-1 states in $G$ and $\wt{G}$ coincide.

\item For $s\in \wt{\SB}$ we have the following cases:
(i)~if $s \in \SB$ (i.e., the state is also a player-2 state in $G$, thus it
is not in $\Source(\livenv)$), then $\wt{E}(s)=E(s)$, i.e, then the set of
edges are same as in $G$; and
(ii)~else $s=\wt{s'}$ and $s' \in \Source(\livenv)$ and 
$E(s') \setminus \livenv \neq \emptyset$, and in this case 
$\wt{E}(s) = E(s') \setminus \livenv$.

\item For a state $s \in \wt{\SR}$ we have the following two sub-cases:
(i)~if $E(s) \subseteq \livenv$, then $\wt{E}(s)=E(s)$ and 
the transition function chooses all states in $E(s)$ uniformly at random;
(ii)~else $\wt{E}(s) = E(s) \cup \set{\wt{s}}$, and the transition 
function is uniform over its successors. 
\end{enumerate}
Intuitively, the edges and transition function can be described as follows:
all states $s$ in $\Source(\livenv)$ are converted to probabilistic states, 
and from states in $\Source(\livenv)$ all edges in $\livenv \cap E(s)$ 
are chosen uniformly at random and also the state $\wt{s}$ 
which is copy of $s$ is chosen from where player~2 has the choice of 
the edges from $E(s)$ that are not contained in $\livenv$.
\end{enumerate}
Given the parity function $p$, we construct the parity function
$\wt{p}$ on $\wt{S}$ as follows:
for all states $s\in S$ we have $\wt{p}(s)=p(s)$, and for a 
state $\wt{s}$ in $\wt{S}$, let $\wt{s}$ be a copy of $s$, 
then $\wt{p}(\wt{s})=p(s)$.
We refer to the above reduction as the edge assumption reduction and
denote it by $\AssRed$, i.e., 
$(\wt{G},\wt{p})=\AssRed(G,\livenv,p)$.
The following theorem states the connection about winning in $G$ for the
objective $\AssumeLive(\livenv,\Parity{p})$ and winning almost-surely
in $\wt{G}$ for $\Parity{\wt{p}}$. 
The key argument for the proof is as follows. 
A memoryless almost-sure winning strategy $\wt{\straa}$ 
in $\wt{G}$ can be fixed in $G$, and it can be shown that the strategy in $G$
is sure winning for the Rabin objective that can be derived from 
the objective $\AssumeLive(\livenv,\Parity{p})$.
Conversely, a memoryless sure winning strategy in $G$ for the 
Rabin objective derived from $\AssumeLive(\livenv,\Parity{p})$ 
can be fixed in $\wt{G}$, and it can be shown that the 
strategy is almost-winning for $\Parity{\wt{p}}$ in 
$\wt{G}$.
A key property useful in the proof is as follows: for a 
probability distribution $\mu$ over a finite set $A$ 
that assigns positive probability to each element in $A$, 
if the probability distribution $\mu$ is sampled infinitely
many times, then every element in $A$ appears infinitely often
with probability~1.

\begin{theorem}\label{thrm:reduction-edgeassumption}
Let $G$ be a deterministic game graph, with a parity objective $\Phi$ defined
by a parity function $p$.
Let $\livenv \subseteq E_2$ be a subset of player-2 edges, and 
let  $(\wt{G},\wt{p}) =\AssRed(G,\livenv,p)$.
Then $\waa(\Parity{\wt{p}}) \cap S= \was(\AssumeLive(\livenv,\Phi))$.
\end{theorem}

Theorem~\ref{thrm:reduction-edgeassumption} presents a linear-time 
reduction for $\AssumeLive(\livenv,\Parity{p})$ to probabilistic games
with parity objectives. 
Using the reduction of Theorem~\ref{thrm:almost} and the results
for deterministic parity games (Theorem~\ref{thrm:sure}) we 
obtain the following corollary.

\begin{corollary}
Given a deterministic game graph $G$, an objective $\Phi$, 
a set $\livenv$ of edges, and a state $s$, 
whether $s \in \was(\AssumeLive(\livenv,\Phi))$ can
be decided in quadratic time if $\Phi$ is a B\"uchi or a coB\"uchi
objective, and in NP $\cap$ coNP if $\Phi$ is a  parity 
objective.
\end{corollary}

\smallskip\noindent{\bf Complexity of computing a minimal strongly fair
  assumptions.}
We now discuss the problem of finding a minimal set of edges on which
a strong fair assumption is sufficient.  Given a deterministic game graph
$\gamegraph$, a B\"uchi objective $\Phi$, a number $k\in \mathbb{N}$,
and a state $s$, we show that 3SAT can be reduced to the problem of
deciding if there is a strongly fair assumption $\livenv$ with at most
$k$ edges ($|\livenv|\le k)$ that is sufficient for $s$ and~$\Phi$.
  
Given a CNF-formula $f$, we will construct a deterministic game graph
$\gamegraph$, give a B\"uchi objective $\Phi$, an initial state $s$,
and a constant $k$, such that $f$ is satisfiable if and only if there 
exists a strongly
fair assumption $\livenv$ of size at most $k$ that is sufficient for
$s$ and $\Phi$.
In Figure~\ref{fig:np-idea} we show a sketch of how to construct $G$:
for each variable $v_i$ we build two player-2 states, one with the
positive literal $l_i$ and one with the negative literal $\bar{l_i}$.
Each state has an edge to a B\"uchi state $B$ and to non-B\"uchi state
$\bar{B}$.  Furthermore, for each variable we add a player-1 state
$v_i$ that connects the two states $l_i$ and $\bar{l_i}$ representing
the literals.  Similarly, for each clause $c_i=l_{i_1}\vee l_{i_2}
\vee l_{i_3}$ we have one player-1 state $c_i$ connected to the state
representing the literals $l_{i_1}$,$l_{i_2}$, and $l_{i_3}$.  Let $n$
be the number of variables and $c$ be the number of clauses in $f$,
and let $j=n+c$.  Starting from the initial state $11$ we have grid of
player-2 states with $j$ columns and from $2$ up to $j$ lines
depending on the column.  A grid state is connected to its right and
to its upper neighbor.  Each grid states in the last column is
connected either to a state $v_i$ representing a variable or a clause
state $c_i$.
The constructed game $\game$ has $3n + c + \frac{(j+1)(j+2)}{2}$
states and $6n + 2c + (j+1)^2$ edges.  We set $\Phi=\Buchi{\set{B}}$,
$s=11$, and $k=n$.

Given a satisfying assignment $I$ for $f$, we build a strongly fair
assumption $\livenv$ that includes for each variable $v_i$ an edge
such that if $I(v_i)=\true$, then $(l_i,B) \in \livenv$, else
$(\bar{l_i},B) \in \livenv$. ($|\livenv|=n=k$.)  The memoryless
strategy $\alpha$ that sets $\alpha(v_i)$ to $l_i$ if $I(v_i)=\true$,
otherwise to $\bar{l_i}$, and $\alpha(c_i)$ to the state representing
the literal that is satisfied in $c_i$ w.r.t.\ $I$., is a winning
strategy for player~1 for state $11$ and the objective
$\AssumeLive(\livenv, \Phi)$. It follow that $\livenv$ is sufficient
for $11$ and $\Phi$.
For the other direction, we observe that any assumption $\livenv$ of
size smaller or equal to $k$ that includes edges from grid states is
not sufficient for state $11$ and $\Phi$.
Assume that there is some grid edge $(s,t)$ in $\livenv$.  Since
$|\livenv|\le k$ there is some variable $v$ for which neither the edge
$(l,B)$ nor $(\bar{l},B)$ is in $\livenv$.  Due to the structure of
the grid, player~2 can pick a strategy that results in plays that
avoids all except for one player-1 states.
From this state $v$, player~1 has the two choices to go to $l$ or
$\bar{l}$ but from both states player~2 can avoid $B$.  Player~2 has a
winning strategy from $11$ and so $\livenv$ is not sufficient.  An
assumption of size $k$ that is sufficient for $11$ and $\Phi$ only
includes edges from literal states $l_i$ or $\bar{l_i}$.
Given an assumption $\livenv$ of size $k$ that is sufficient for $11$
and $\Phi$.
Since $\livenv$ include only edges from literal states, we can easily
map $\livenv$ to a satisfying assignment $I$ for $f$: If $(l_i,B) \in
\livenv$, then $I(v_i)=\true$, and if $(\bar{l_i},B) \in \livenv$,
then $I(v_i)=\false$.

\begin{figure}[bt]
  \begin{minipage}[t]{0.40\textwidth}
    \centerline{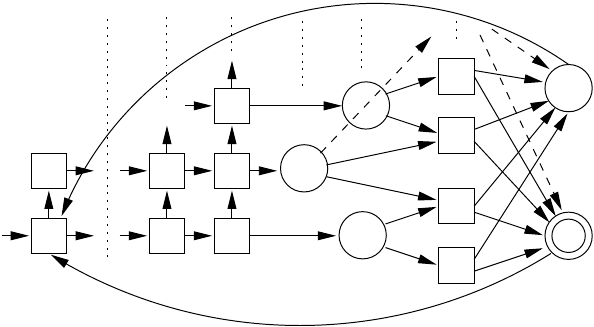} 
    \caption{Idea of the NP-hardness proof.}
    \label{fig:np-idea}
  \end{minipage}
  \begin{minipage}[t]{0.34\textwidth}
    \centerline{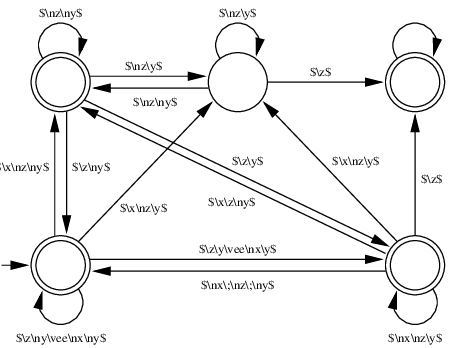}
    \caption{Constructed environment assumption for the specification
      $\always(\req\implies\eventually\grant)$ $\wedge$
      $\always(\cancel\implies\nextt\neg\grant)$.}
    \label{fig:lang}
  \end{minipage}
  \hspace{0.02\textwidth}
  \begin{minipage}[t]{0.22\textwidth}
    \centerline{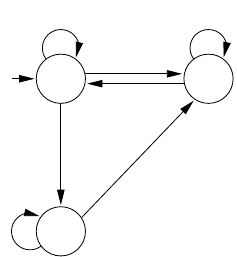}
    \caption{System constructed with assumption shown in Figure~\ref{fig:lang}.}
    \label{fig:system}
  \end{minipage}
\end{figure}

\begin{theorem}
  Given a deterministic game graph $\gamegraph$, a B\"uchi objective
  $\Phi$, a number $k\in \mathbb{N}$, and a state $s$, deciding if
  there is a strongly fair assumption $\livenv$ with at most $k$ edges
  ($|\livenv|\le k$) that is sufficient for $s$ and $\Phi$ is NP-hard.
\end{theorem}

\medskip\noindent{\bf 
Computing locally-minimal strongly fair assumptions.}
Since finding the minimal set of edges is NP-hard, we focus on computing a 
\emph{locally} minimal set of edges.
Given a deterministic game graph $G$, a state $s \in S$, and a parity
objective $\Phi$, we call a set $\livenv \subseteq E_2$ of player-2
edges locally-minimal strongly fair assumption if $s \in
\was(\AssumeLive(\livenv,\Phi))$ and for all proper subsets $\livenv'$ of
$\livenv$ we have $s \not \in \was(\AssumeLive(\livenv',\Phi))$.  We now
show that a locally-minimal strongly fair assumption set $\livenv^*$ of
edges can be computed by polynomial calls to a procedure that checks
given a set $\livenv$ of edges whether $s \in
\was(\AssumeLive(\livenv,\Phi))$.  The procedure is as follows:
\begin{enumerate}
\item \emph{Iteration $0$.} Let the initial set of assumption edges be
  all player-2 edges, i.e., let $E^0=E_2$; if $s \not\in
  \was(\AssumeLive(E^0,\Phi))$, then there is no subset $\livenv$ of
  edges such that $s \in \was(\AssumeLive(\livenv,\Phi))$.  If $s \in
  \was(\AssumeLive(E^0,\Phi))$, then we proceed to the next iterative
  step.

\item \emph{Iteration $i$.} Let the current set of assumption edges be
  $E^i$ such that we have $s \in \was(\AssumeLive(E^i,\Phi))$.  If
  there exists $e \in E^i$, such that $s \in \was(\AssumeLive(E^i
  \setminus \set{e}, \Phi))$, then let $E^{i+1} = E^i \setminus
  \set{e}$, and proceed to iteration $i+1$.  Else if no such $e$
  exists, then $E^*=E^i$ is a locally-minimal strongly fair assumption set
  of edges.
\end{enumerate}
The claim that the set of edges obtained above is a locally-minimal
strongly fair assumption set can be proved as follows: for a set $\livenv$
of player-2 edges, if $s\not\in\was(\AssumeLive(\livenv,\Phi))$, then
for all subsets $\livenv'$ of $\livenv$ we have
$s\not\in\was(\AssumeLive(\livenv',\Phi))$.  It follows from above that
for the set $E^*$ of player-2 edges obtained by the procedure
satisfies that $s\in \was(\AssumeLive(E^*,\Phi))$, and for all proper
subsets $E'$ of $E^*$ we have $s \not\in \was(\AssumeLive(E',\Phi))$.
The desired result follows.

\begin{theorem}
  \label{thm:compute_liveness}
  The computed set $E^*$ of edges is a locally-minimal strongly
  fair assumption.
\end{theorem}




\section{Combining Safety and Liveness Assumptions}
\label{sec:combining}

Let $\varphi$ be a specification and let
$\game_\varphi=(\gamegraph,\init,\lambda)$ be the corresponding
synthesis game with winning objective $\Phi$.  We first compute a
non-restrictively safety assumption $\safenv$ as described in
Section~\ref{sec:safety}.  If $\varphi$ is satisfiable, it follows
from Theorem~\ref{thm:env-realizable} and~\ref{thm:compute_safety}
that $\safenv$ exists and that the corresponding environment
assumption is realizable for the environment.
Then, we modify the winning objective of player~1 with the computed
safety assumption: we extend the set of winning plays of player~1 with
all plays, in which player~2 follows one of the edges in $\safenv$.
Since $\safenv$ is safe-sufficient, it follows that $\init$ is live
for player~1 in the modified game.
On the modified game, we compute a locally-minimal strongly fair
assumption as described in Section~\ref{sec:liveness}
(Theorem~\ref{thm:compute_liveness}).  Finally, using
Theorem~\ref{thm:complete} and \ref{thm:liveness}, we conclude the
following.

\begin{theorem}
  Given a specification $\varphi$, if the assumption $\psi =
  \psi_\safenv \cap \psi_\livenv \ne \emptyset$, where $\safenv$ and
  $\livenv$ are computed as shown before, then $\psi$ is a sufficient
  assumption for $\varphi$ that is realizable for the environment.
  If $\varphi$ has a corresponding synthesis game $\game_\varphi$ with
  a safety, reachability, or B\"uchi objective for player~1, then if
  there exists a sufficient environment assumption $\psi\ne
  \emptyset$, then the assumption $\psi = \psi_\safenv \cap
  \psi_\livenv$, where $\safenv$ and $\livenv$ are computed as shown
  before, is not empty.
\end{theorem}

Recall the example from the introduction with the signals $\req$,
$\cancel$, and $\grant$ and the specification
$\always(\req\implies\nextt\eventually\grant)$ $\wedge$
$\always((\cancel\vee\grant)\implies\nextt\neg\grant)$.  Applying our
algorithm we get the environment assumption $\wh{\psi}$ shown in
Figure~\ref{fig:lang} (double lines indicate B\"uchi states).  We
could not describe the language using an LTL formula, therefore we
give its relation to the assumptions proposed in the introduction.
Our assumption $\wh{\psi}$ includes $\psi_1=\always(\neg\cancel)$ and
$\psi_2=\always(\eventually(\neg\cancel))$, is a strict subset of
$\psi_6 = \xi \weakuntil (\xi\wedge(\cancel\vee
\grant)\wedge\nextt\grant)$ with $\xi=
\req\implies\nextt\eventually(\neg\cancel \vee \grant)$, and is
incomparable to all other sufficient assumptions.
Even though, the constructed assumption is not the weakest w.r.t.\
language inclusion, it still serves its purpose:
Figure~\ref{fig:system} shows a system synthesized with a modified
version of \cite{Jobstm06c} using the assumption $\wh{\psi}$.

\section{Acknowledgements}

The authors would like to thank Rupak Majumdar for stimulating and
fruitful discussions.




\bibliographystyle{plain}
\bibliography{main}
\end{document}